\RequirePackage[]{silence}
\documentclass[11pt]{article}

\usepackage[letterpaper,margin=1.00in]{geometry}
\usepackage{amsmath, amssymb, amsthm, amsfonts}
\usepackage{yfonts, mathrsfs}
\usepackage{thm-restate}
\usepackage[textsize=tiny]{todonotes}

\usepackage{ifthen}
\usepackage{arrayjobx}
\usepackage{tikz}
\usetikzlibrary{positioning,decorations.pathreplacing}

\usepackage{cite}
\usepackage{appendix}
\usepackage{graphicx}
\usepackage{color}
\usepackage{algorithm}
\usepackage[noend]{algpseudocode}
\usepackage{epstopdf}
\usepackage{subfig}
\usepackage{wrapfig}
\usepackage[framemethod=tikz]{mdframed}
\usepackage[bottom]{footmisc}
\usepackage{enumitem}
\setitemize{noitemsep,topsep=0pt,parsep=0pt,partopsep=0pt}
\usepackage{caption}
\usepackage{xspace}
\usepackage{hyperref}
\definecolor{darkgreen}{rgb}{0,0.5,0}
\hypersetup{
    unicode=false,          
    colorlinks=true,        
    linkcolor=red,          
    citecolor=darkgreen,        
    filecolor=magenta,      
    urlcolor=cyan           
}
\usepackage[capitalize, nameinlink]{cleveref}
\crefformat{footnote}{#2\footnotemark[#1]#3}
\algnewcommand\algorithmicswitch{\textbf{switch}}
\algnewcommand\algorithmiccase{\textbf{case}}

\algdef{SE}[SWITCH]{Switch}{EndSwitch}[1]{\algorithmicswitch\ #1\ \algorithmicdo}{\algorithmicend\ \algorithmicswitch}%
\algdef{SE}[CASE]{Case}{EndCase}[1]{\algorithmiccase\ #1}{\algorithmicend\ \algorithmiccase}%
\algtext*{EndSwitch}%
\algtext*{EndCase}%

\newcommand{\eps}{\varepsilon}
\newcommand{\congest}{$\mathsf{CONGEST}$\xspace}

\newcommand{\poly}{\operatorname{\text{{\rm poly}}}}
\newcommand{\floor}[1]{\lfloor #1 \rfloor}

\newcommand{\Parts}{\mathcal{P}}

\renewcommand{\paragraph}[1]{\vspace{0.15cm}\noindent {\bf #1}:}
\newtheorem{theorem}{Theorem}[section]

\newtheorem{lemma}[theorem]{Lemma}

\newtheorem{corollary}[theorem]{Corollary}

\newtheorem{observation}[theorem]{Observation}
\newtheorem{definition}[theorem]{Definition}

\newcommand{\FullOrShort}{full}

\ifthenelse{\equal{\FullOrShort}{full}}{
	
  \newcommand{\fullOnly}[1]{#1}
  \newcommand{\shortOnly}[1]{}
  
  }{

    \newcommand{\fullOnly}[1]{}
    \newcommand{\shortOnly}[1]{#1}

  }

\begin{document}

\date{}

\title{Low-Congestion Shortcuts for Graphs Excluding Dense Minors}

\author{
  Mohsen Ghaffari\thanks{Supported in part by funding from the European Research Council (ERC) under the European Union’s Horizon 2020 research and innovation programme (grant agreement No. 853109), and the Swiss National Foundation (project grant 200021-184735).
  }\\
  \small ETH Zurich \\
  \small ghaffari@inf.ethz.ch
\and	
	Bernhard Haeupler\thanks{Supported in part by NSF grants CCF-1814603, CCF-1910588, NSF CAREER award CCF-1750808 and a Sloan Research Fellowship.} \\
  \small Carnegie Mellon University\\
  \small haeupler@cs.cmu.edu
 }

\maketitle

\setcounter{page}{0}
\thispagestyle{empty}

\begin{abstract}
We prove that any $n$-node graph $G$ with diameter $D$ admits shortcuts with congestion $O(\delta D \log n)$ and dilation $O(\delta D)$, 
where $\delta$ is the maximum edge-density of any minor of $G$. Our proof is simple, elementary, and constructive -- featuring a $\tilde{\Theta}(\delta D)$-round\footnote{We use $\tilde{O}$-notation to suppress polylogarithmic factors in $n$, e.g., $\tilde{O}(f(n,D,\delta)) = O(f(n,D,\delta) \log^{O(1)} n)$.} distributed construction algorithm. Our results are tight up to $\tilde{O}(1)$ factors and generalize, simplify, unify, and strengthen several prior results. For example, for graphs excluding a fixed minor, i.e., graphs with constant $\delta$, only a $\tilde{O}(D^2)$ bound was known based on a very technical proof that relies on the Robertson-Seymour Graph Structure Theorem. 

\medskip
A direct consequence of our result is that many graph families, including any minor-excluded ones, have near-optimal $\tilde{\Theta}(D)$-round distributed algorithms for many  fundamental communication primitives and optimization problems including minimum spanning tree, minimum cut, and shortest-path approximations.

\end{abstract}
\setcounter{page}{0}
\thispagestyle{empty}

\bigskip
{   
    \bigskip
    \hypersetup{linkcolor=blue}
    \tableofcontents
}

\newpage

\section{Introduction and Related Work}
\textit{Low-congestion shortcuts} are graph-theoretic objects whose quality captures the distributed complexity of a wide range of fundamental graph problems and communication primitives. In this paper, we provide nearly-tight shortcuts for all graphs excluding (dense) minors. Our results significantly strengthen, simplify, generalize and unify several prior results on shortcuts for restricted graph classes. Our results also directly imply simple near-optimal distributed algorithms for a number of well-studied graph problems in excluded minor graphs and many other graph families. 

\subsection{Background and Definition}
\paragraph{Model} We work with the standard synchronous message passing model of computation on networks and distributed systems. The communication network is abstracted as an $n$-node undirected graph $G=(V, E)$ with $m =|E|$ edges. In each communication round, each node can send an $O(\log n)$-bit message to each of its neighbors in $G$. At the beginning, nodes do not know the topology of $G$. At the end, each node should know its own part of the output, e.g., which of its edges are in the computed \textit{minimum spanning tree}. This model is sometimes referred to as the \congest model~\cite{Peleg:2000}.

\paragraph{The Motivation for Shortcuts}
The $\tilde{\Omega}(\sqrt{n})$-round lower bound for global distributed graph problems is well-known by now. Concretely, it is known that on worst-case general graphs, for a strikingly wide-range of global graph problems --- including minimum spanning tree, minimum cut, maximum flow, single-source shortest path, etc --- any distributed algorithm needs a round complexity of $\tilde{\Omega}(\sqrt{n})$ ~\cite{DasSarma-11, Elkin-2004, Peleg-Rubinovich-1999}. This holds even for any (non-trivial) approximation of these problems, and even on graphs with a small, e.g., logarithmic, diameter. On the other hand, the lower bound graphs are carefully crafted pathological topologies which do not occur in practice. Indeed, many real-world networks have small (polylogarithmic) diameters (e.g., the Facebook graph with billions of nodes has average distance below $5$ and a diameter of around $50$) and seem to allow for exponentially faster optimization algorithms with $\tilde{O}(D)$-round complexities. Low-congestion shortcuts~\cite{GH-PlanarII} were introduced as a graph-theoretic notion to capture and exploit this phenomenon and allow a more fine-grained study of the complexity of global problems and how this complexity relates to the structure of the network topology. In various graph families algorithms based on low-congestion shortcuts have gone well below the $\tilde{\Omega}(\sqrt{n})$ barrier often obtaining near-optimal $\tilde{O}(D)$-round algorithms\footnote{One can draw parallels between the role that shortcuts have turned out to play for distributed algorithms of global graph problems with the role that \textit{separators} play as a key algorithmic tool in sequential algorithms for various graph families, such as planar~\cite{LiptonTarjan,miller1984finding}, bounded genus~\cite{gilbert1984separatorGenus}, and minor-excluded~\cite{alon1990separator}.}.

\paragraph{The Definition of Shortcuts} Suppose that the set $V$ of vertices is partitioned into disjoint subsets $V_1$, $V_2$, \dots, $V_k$, known as \textit{parts}, such that the subgraph $G[V_i]$ induced by each part $V_i$ is connected. We call a collection of subgraphs $H_1$, $H_2$, \dots, $H_k$ a \textit{shortcut with congestion $c$ and dilation $d$} if we have the following two properties: (A) $\forall i\in [1, k]$, the diameter of the subgraph $G[V_i]+H_i$ is $O(d)$, and (B) each edge $e\in E$ is in $O(c)$ many subgraphs $H_i$. 

\paragraph{Application of Shortcuts} Shortcuts organically lead to fast distributed algorithms. For instance, let us consider the minimum spanning tree problem. Suppose that we are in a graph family for which, we have a distributed algorithm that can compute a shortcut with congestion $c$ and dilation $d$ in $T$ rounds (upon receiving the partition $V_1$, $V_2$, \dots, $V_k$). In general, we refer to $c+d$ as the \textit{quality of the shortcut}. Then, we can use this algorithm to obtain a distributed algorithm for the minimum spanning tree problem, with a round complexity $\tilde{O}(c+d+T)$. This follows directly from Boruvka's 1926 approach~\cite{Boruvka26}. A number of other graph problems can also be solved distributedly using shortcuts, with a similar complexity. This includes 
min-cut computation~\cite{GH-PlanarII}, single-source shortest-paths approximation~\cite{haeupler2018faster}, and many more~\cite{ghaffari2017DFS,ghaffari2017DFS,dory2019improved,li2019planar}.

Of course, the question that remains is this: What is the existential shortcut quality $c+d$ of important graph families, and what is the corresponding construction time $T$?

\paragraph{Graph Minors and Minor Density} 
Before proceeding to known results and our contribution, let us briefly recall the definition of minors and their density.

A graph $H$ is a minor of graph $G$ if $H$ can be obtained from $G$ by contracting edges and deleting edges and vertices. Equivalently $H=(V',E')$ is a minor of graph $G=(V,E)$ if there is a mapping $\text{map}_{H,G}$ from vertices in $H$ to disjoint connected subsets of vertices in $G$, each inducing a connected subgraph, such that for every edge $\{u',v'\} \in E'$ there exists a $\{u,v\} \in E$ with $u \in \text{map}_{H,G}(u')$ and $v \in \text{map}_{H,G}(v')$.

An important parameter throughout this paper is the \textit{minor density} $\delta(G)$ of a graph $G$  which is defined as:
$$\delta(G) = \max\left\{\frac{|E'|}{|V'|} \,\,\, \middle\vert \,\,\,  H = (V',E') \text{ is a minor of } G\right\}.$$ It is known that the minor density $\delta(G)$ is (up to a small polylogarithmic factor) the same as the size of the largest complete minor in $G$, i.e., its \textit{complete-graph minor size} $r(G) = \max\{r \mid K_r \text{ is a minor of } G\}$. 

\begin{lemma}\cite{thomason2001extremal}
$\forall G: \ \frac{r(G)-1}{2} \leq \delta(G) \leq 8 r(G) \sqrt{\log_2 r(G)}$, i.e., $\delta(G) = \tilde{\Theta}(r(G))$. 
\end{lemma}

Note if a graph family $\mathcal{G}$ is closed under taking minors or equivalently excludes a fixed minor $H$ of size $s$ then very $G \in \mathcal{G}$ has $r(G) < s$ and therefore also a constant minor density $\delta(G) = O(s\sqrt{\log s})$.

\subsection{Our Contribution}\label{sec:OurResults} We first briefly summarize known results on shortcuts (see \Cref{subsec:priorWork} for further details). Ghaffari and Haeupler~\cite{GH-PlanarII, GH-Embedding} provided shortcuts of quality $\tilde{O}(D)$ and construction time $\tilde{O}(D)$ for planar graphs. This was later extended to graphs with bounded genus, bounded treewidth, and bounded pathwidth~\cite{haeupler2016treewidth}, and with improved construction algorithms~\cite{haeupler2016low,haeupler2018round}. Haeupler, Li, and Zuzic~\cite{haeupler2018minor} gave shortcuts for excluded minor graphs, with quality $\tilde{O}(D^2)$, using elaborate arguments building on the Graph Structure Theorem of Robertson and Seymour~\cite{robertson1986graph,robertson2003graph}.
While excluded minor graphs are vastly more general, encompassing all previous graph classes for constant bounds on the above graph parameters, the question whether near-optimal $\tilde{O}(D)$ shortcuts and optimization algorithms for excluded minor topologies are possible remained open.  

\paragraph{Existential Results for Graph with Minor Density $\delta$}  In this work, we resolve this question in the positive. We also significantly strengthen, simplify, generalize, and unify all the above results by giving an near-optimal existential shortcut guarantee for \emph{any} graph $G$, which depends only on its minor density $\delta(G)$:

\begin{theorem}
\label{thm:MainExistentialDilation} Any $n$-node graph $G$ with diameter $D$ and minor density $\delta(G)=\delta$ admits shortcuts with dilation $O(\delta D)$ and congestion $O(\delta D \log n)$.
\end{theorem}

Besides the important quadratic quantitative improvement from $\tilde{O}(D^2)$ to a near-optimal $\tilde{O}(D)$ for excluded minor graphs, i.e., graphs with constant $\delta$, our proof is significantly simpler than that of~\cite{haeupler2018minor}. Instead of the technical proof in \cite{haeupler2018minor} which uses the powerful Graph Structure Theorem~\cite{robertson1986graph,robertson2003graph}, we give a short, elegant, and elementary proof. Our proof even provides small concrete constants\footnote{No attempt was made to optimize the explicit constants in the $8 \delta$ and $8 \delta D$ bounds of our main result \Cref{thm:Main}, or any other explicit constants in this paper. Somewhat better constants are likely possible.} whereas the Robertson-Seymour Graph Structure Theorem is known to often lead to tower-type  dependencies on the size $k$ of the excluded minor~\cite{lokshtanov2020efficient} and ``galactic algorithms''~\cite{lipton2013people}. 

Even more importantly, our result applies to \emph{any} graph and graph family even if $\delta(G)$ is large or growing. For example it implies that graphs with sub-polynomial minor density or expansion (see \cite{nevsetvril2012sparsity} for definitions and treatment of such more inclusive graph families) have sub-polynomial shortcuts and fast $O(n^{O(1)})$-round optimization algorithms, which still drastically improve over the $\tilde{\Omega}(\sqrt{n})$-lower bound. Our results are the first that apply to graph classes strictly more general than minor-closed or excluded-minor graph families.

We complement \Cref{thm:MainExistentialDilation} with a matching lower bound which shows that the linear dependency of \Cref{thm:MainExistentialDilation} on the minor-density $\delta(G)$ is necessary and optimal: 

\begin{lemma}\label[lemma]{lem:tightnessOfDensityDependecy}
For any $\delta,D$ there is a graph $G$ with diameter $D$ and $\delta(G) \leq \delta$ and a collection of parts for which the quality of the best shortcut in $G$ is $\Omega(\delta D)$.
\end{lemma}

Lastly, having the minor density $\delta(G)$ as a parameter in \Cref{thm:MainExistentialDilation} also directly implies often optimal results on shortcuts for other graph parameters in a simple and uniform way. For example, \cite{GH-PlanarII} showed that any genus-$g$ network admits shortcuts of quality $\tilde{O}(gD)$. In~\cite{haeupler2016treewidth} this dependency of the shortcut quality on the genus was improved to an optimal bound of $\tilde{\Theta}(\sqrt{g}D)$ via sophisticated topological arguments including cutting a genus-$g$ graph along fundamental cycles of a shortest-path tree and cleverly combining planar shortcuts for pieces of parts that are cut by these cycles. We obtain the same result as a trivial corollary of \Cref{thm:MainExistentialDilation}, given that $\delta(G)=O(\sqrt{q})$ for any genus-$g$ graph.

\begin{corollary}\label{cor:genus}
Any genus-$g$ graph with $n$ nodes and diameter $D$ admits shortcuts with congestion $O(\sqrt{g} D \log n)$ and dilation $O(\sqrt{g}D)$. 
\end{corollary}

Bounds for other graph parameters follow similarly as simple corollaries of our main theorem. This includes tight $\tilde{O}(kD)$ shortcuts for $k$-pathwidth and $k$-treewidth graphs, matching the results in \cite{haeupler2016treewidth} (since $\delta(G) = O(k)$ for such graphs). In contrast to prior works like \cite{haeupler2016treewidth}, we do not require a completely different proof specific to the graph parameter at hand to obtain these bounds.

\paragraph{Distributed Construction and Applications} 
In order to be algorithmically useful we also need fast distributed constructions for the new existential shortcut guarantees. We achieve this by proving strong additional structural guarantees for our shortcuts, namely tree-restrictedness and a small block number. Haeupler, Hershkovitz, Izumi, and Zuzic~\cite{haeupler2016low,haeupler2018round} showed that these structural guarantees are strong enough for a simple uniform shortcut constructions to find a $\tilde{O}(Q)$-quality shortcut in only $\tilde{O}(Q)$ rounds whenever a quality-$Q$ shortcut with such structure exists.  

\begin{theorem}\label{thm:mainconstructive}
There exists a randomized distributed algorithm which, for any $n$-node $m$-edge graph $G$ with diameter $D$ and minor density $\delta$, computes a shortcut of quality $\tilde{O}(\delta D)$ in $\tilde{O}(\delta D)$ rounds with high probability. There also exists a $\tilde{O}(\delta^2 D)$-round deterministic algorithm. Both algorithms use only $\tilde{O}(m)$ messages.
\end{theorem}

This, together with all the algorithms that are built on top of the low-congestion shortcut framework, shows that a wide range of fundamental communication primitives and global graph problems --- e.g., minimum spanning tree, min-cut, shortest path approximation, etc --- can be solved in $\tilde{O}(D\delta)$ rounds, in graphs that do not have a minor of density $\delta$. This vastly widens the range of graph families for which we now know the correct round complexity for basic global graph problems (up to logarithmic factors). We state just two such corollaries as examples. 

\begin{corollary}\label{cor:MST}
There is a distributed algorithms that, for any $n$-node $m$-edge graph $G$ with diameter $D$ and minor density $\delta$, computes a minimum spanning tree in $\tilde{O}(\delta D)$ rounds with high probability (or $\tilde{O}(\delta^2 D)$ rounds deterministically) using $\tilde{O}(m )$ messages.
\end{corollary}

\begin{corollary}\label{cor:MinCut}
There is a randomized distributed algorithms that, for any $n$-node $m$-edge graph $G$ with diameter $D$ and minor density $\delta$, computes an exact minimum cut in $G$ in $\tilde{O}(\delta^{O(1)} D)$ rounds with high probability using $\tilde{O}(m)$ messages.
\end{corollary}

Similar results can be obtained for for sub-graph connectivity, single-source shortest paths approximations~\cite{haeupler2018faster}, as well as several fundamental communication primitives like multiple unicasts or partwise-aggregation (see \Cref{sec:shortcutDefinitions}).

\subsection{Prior Work on Shortcuts}
\label{subsec:priorWork}
We next overview the known results about the existential quality and construction time of shortcuts.

Let us start with general graphs. It is easy to see that any $n$-node graph $G$, whose diameter is at most $D$, admits shortcuts of quality $D+\sqrt{n}$: Let $T$ be a BFS of the graph $G$. Define $H_i=\emptyset$ for any each part with $|V_i|\leq \sqrt{n}$ and $H_i = T$ for any other part. Moreover, this $D+\sqrt{n}$ bound is nearly-optimal~\cite{GH-PlanarII} in general graphs. This $D+\sqrt{n}$ bound is implicitly the underlying reason for the seminal $\tilde{O}(D+\sqrt{n})$ round minimum spanning tree algorithm of Kutten and Peleg~\cite{Kutten-Peleg, Garay-Kutten-Peleg}, although there are more fine-grained aspects there to avoid extra logarithmic factors. Moreover, this $\tilde{O}(D+\sqrt{n})$ round complexity is nearly optimal for solving MST in general graphs, as mentioned before, because of the carefully crafted lower bound graphs of \cite{DasSarma-11, Elkin-2004, Peleg-Rubinovich-1999}.

Ghaffari and Haeupler~\cite{GH-PlanarII} showed that any planar graph with diameter at most $D$ admits shortcuts of quality $O(D\log D)$, and that this bound is nearly tight, almost matching a lower bound of $\Omega(D \, \frac{\log D}{\log\log D})$. The also provided a distributed algorithm for constructing such shortcuts in $\tilde{O}(D)$ rounds, assuming a planar embedding of the graph. Combined with the distributed planar embedding algorithm provided by ~\cite{GH-Embedding}, this led to an $\tilde{O}(D)$ round distributed algorithm for MST, and some other graph problems, thus exhibiting a family of graphs in which one can go below the notorious $\tilde{\Omega}(D+\sqrt{n})$ lower bound~\cite{DasSarma-11}. They~\cite{GH-PlanarII} also showed that graphs embeddable on a surface of genus $g$ admit shortcuts of quality $O(gD\log D)$.

Haeupler, Izumi, and Zuzic~\cite{haeupler2016low} showed that one can construct shortcuts with a quality similar to those of \cite{GH-PlanarII}, even without a planar embedding. Haeupler, Izumi, and Zuzic~\cite{haeupler2016treewidth} showed that graphs of treewidth or pathwidth $k$ admit shortcuts of quality $\tilde{O}(kD)$, and they also improved the genus dependency to $\tilde{O}(\sqrt{g}D)$. Haeupler, Hershkowitz, and Wajc~\cite{haeupler2018round} improved and extended the results of ~\cite{haeupler2016low} by showing that one can obtain algorithms that are near-optimal in both time and message complexity, using shortcuts.

Generalizing the span of shortcuts much further, Haeupler, Li, and Zuzic~\cite{haeupler2018minor} gave shortcuts for excluded minor graphs. Concretely, they showed any graph that does not have a clique of constant size as a minor admits a shortcut of quality $\tilde{O}(D^2)$. This result is fairly involved and it builds on the Graph Structure Theorem of Robertson and Seymour~\cite{robertson1986graph,robertson2003graph}. Using the approach of\cite{haeupler2016low,haeupler2018round}, all these shortcuts mentioned above\cite{GH-PlanarII, haeupler2016low, haeupler2016treewidth, haeupler2018minor} can be constructed in a round complexity matching their quality, up to logarithmic factors. 

Ghaffari, Kuhn, and Su~\cite{ghaffari2017mixing} provided shortcut constructions for well-connected graphs. In particular, they showed that any graph where the lazy random walk has mixing time $T_{mix}$ admits shortcuts of quality $T_{mix} \cdot \poly(\log n)$ and they presented distributed algorithms for constructing shortcuts of quality $T_{mix} \cdot 2^{O(\sqrt{\log n\log \log n})}$ in $T_{mix} \cdot 2^{O(\sqrt{\log n\log \log n})}$ rounds. These quality and construction time bounds were both improved later to $T_{mix} \cdot 2^{O(\sqrt{\log n})}$, by Ghaffari and Li~\cite{ghaffari2018mixing}. 

Finally, Kitamura et al.\cite{kitamura2019low} showed that any $k$-chordal graph admits shortcuts of quality $O(kD)$, which can be determined in even $O(1)$ rounds, and that graph of diameter $4$ and $3$ admit shortcuts of quality and construction time $\tilde{O}(n^{1/4})$ and $\tilde{O}(n^{1/3})$ respectively. These essentially match the lower bounds for Lotker, Patt-Shamir, and Peleg~\cite{lotker2006distributed}. 

\section{Preliminaries: Definitions of Shortcuts}\label{sec:shortcutDefinitions}
In many distributed algorithms for global graph problems, the problem boils down to the following natural part-wise aggregation task. For instance, this exactly captures the problem of identifying minimum-weight outgoing edges in Boruvka's classic approach to MST~\cite{Boruvka26}. 
\begin{definition}(\textbf{The Part-wise Aggregation Problem})
Consider a network graph $G=(V, E)$ and suppose that the vertices are partitioned into disjoint parts $P_1$, $P_2$, \dots, $P_k$ such that the subgraph induced by each part connected. In the part-wise aggregation problem, the input is a value $x_v$ for each node $v\in V$. The output is that each node $u \in P_i$ should learn an aggregate function of the values held by vertices in its part $P_i$, e.g., $\sum_{v\in P_i} x_v$, $\min_{v\in P_i} x_v$, or $\max_{v\in P_i} x_v$. Alternatively, exactly one node in each part has a message and it should be delivered to all nodes of the part.
\end{definition}

\medskip
Typically, if we can solve the part-wise aggregation problem in a network in $T_{PA}$ time, we can obtain algorithms for various fundamental graph problems with a round complexity of $\tilde{O}(D+T_{PA})$. Thus, we want fast algorithms for part-wise aggregation.

They point to emphasize in the part-wise aggregation problem is that, it is possible that the diameter of the subgraph induced by each part $P_i$ is quite large, much larger than the diameter of the base graph $G$. For instance, the former can be up to $\Theta(n)$ while the latter is just $2$ (considering a wheel graph with one part being all nodes except the center). Hence, to obtain fast algorithms for part-wise aggregation, we would like to allow some parts to use edges of $G$ which are outside the part. Of course, we should limit the number of parts that try to use each single given edge, as that would cause congestion and would slow down the solution. This naturally brings us to the concept of \textit{low-congestion shortcuts}, as defined in~\cite{GH-PlanarII}.

\begin{definition}(\textbf{Shortcuts}) Given a part-wise aggregation problem --- i.e., a graph $G=(V, E)$ where vertices are partitioned in disjoint parts $P_1$, $P_2$, \dots, $P_k$, each of which induces a connected subgraph --- we call a collection of subgraph $H_1$, $H_2$, \dots, $H_k$ a shortcut with congestion $c$ and dilation $d$ if we have the following two properties: (I) for each $i\in [1, k]$, the diameter of the subgraph $G[P_i]+H_i$ is at most $O(d)$, and (II) each edge $e$ is in at most $O(c)$ many of the subgraphs $H_i$. We refer to $Q=c+d$ as the \textit{quality of the shortcut}.
\end{definition}

Given a $c$-congestion $d$-dilation shortcut for a part-wise aggregation problem, we can solve the part-wise aggregation problem in $O(c+d\log n)=\tilde{O}(Q)$ rounds, using the \textit{random delays} technique~\cite{leighton1994packet, Ghaffari-Scheduling,haeupler2019nearoptimal}. This makes the shortcut quality $Q$ the dominant parameter which determines the round complexity of shortcut-based algorithms (modulo the shortcut's construction time).

\paragraph{Tree-Restricted Shortcuts, and their Block-Number} In many graph families~\cite{GH-PlanarII,haeupler2016treewidth,haeupler2018minor}, shortcuts can be chosen to come from one low-depth tree of the original graph, which provides a particularly clean and simple structure which can be utilized for efficient shortcut constructions. In particular, one can fix any rooted breadth first search tree $T$--- or any other low-depth spanning tree--- and restrict each $H_i$ to only include edges from $T$. Note that only edges with descendants in $P_i$ are useful for short-cutting $P_i$. A particular simple way to construct $H_i$ is to specify a set $S$ of $b$ edges in $T$ and define $H_i$ to be all edges with descendants in the forest $T  \setminus S$. We say such an $H_i$ has block number at most $b$ since the graph $(P_i \cup V(H_i), H_i)$ has at most $b$ connected components (of diameter $O(D)$). More generally we define tree-restricted shortcuts and their block number as follows: 

\begin{definition} (\textbf{Tree-Restricted Shortcuts and Block Number})
Consider a part-wise aggregation setup in graph $G$ with diameter $D$, parts $P_1$, \dots, $P_k$, a rooted tree $T$ of $G$ with depth $D$. We say the shortcut $H_1, \ldots, H_k$ is tree-restricted or $T$-restricted if all its edges are in $T$, i.e., if $\bigcup_i H_i \subseteq T$. Moreover, for any part $P_i$, we call the connected components of the graph $(P_i \cup V(H_i), H_i)$ the blocks of $P_i$. The block number $b$ of a shortcut is the maximum block number of any part.
\end{definition}

If in a topology has shortcuts for any collection of parts and any choice of the tree $T$ then we say it admits good shortcuts. We also introduce the concept of partial shortcuts which lead to slightly tighter bounds and simpler proofs:

\begin{definition}(\textbf{Admitting Shortcuts})
We say a topology $G$ with diameter $D$ \textit{admits tree-restricted $c$-congestion $b$-block shortcuts}, if for any tree $T$ with depth at most $D$ and for any collection of node-disjoint connected parts $P_1, \ldots, P_k$, there exists a $T$-restricted $c$-congestion $b$-block shortcut. 
\end{definition}

\begin{definition}(\textbf{Admitting Partial Shortcuts}) We say a topology $G$ with diameter $D$ admits tree-restricted  $c$-congestion $b$-block \textit{partial shortcuts}, if for any tree $T$ with depth at most $D$ and for any collection of node-disjoint connected parts $P_1, \ldots, P_k$, there are at least $k/2$ of the parts with a $T$-restricted $c$-congestion $b$-block shortcut.
\end{definition}

It is easy to see that a small block number directly implies a small dilation and that admitting partial shortcuts is essentially the same as admitting shortcuts -- up to a $ \tilde{O}(1)$ factor in the congestion.   

\begin{observation}\label{obs:BlockImpliesDilation}
Any $b$-block $T$-restricted shortcut in a graph with diameter $D$ has dilation at most $b(2D+1)$
\end{observation}
\begin{proof}
For each part $P_i$, the graph $(P_i \cup V(H_i), H_i) \subseteq T$ is a forest with at most $b$ connected components, each of them has a diameter of at most the twice the depth of $T$ which is $D$. These components are connected via edges of $P_i$, because $P_i$ induces a connected subgraph. Hence, the diameter of $G[P_i]+H_i$, i.e., the dilation for this part, is at most $b(2D+1)$. 
\end{proof}

\begin{observation}\label{obs:PartialImpliesFullShortcut}
Any $n$-node graph $G$ that admits tree-restricted $c$-congestion $b$-block partial shortcuts also admits tree-restricted $c \log_2 n$-congestion $b$-block shortcuts.
\end{observation}
\begin{proof}
For a collection of $k$ parts and a tree $T$ consider $\log_2 k$ iterations in which one takes a $T$-restricted $c$-congestion $b$-block partial shortcut for any still remaining parts. Given that such a partial shortcut defines sets of shortcut edges for at least half of the remaining parts, all parts will have a set of shortcut edges from $T$ in the end. Taking this as a full $T$-restricted shortcut might lead to a $(c \log_2 k)$-congestion but leaves the block number of $b$ for every part unaffected. 
\end{proof}

\paragraph{Efficient Shortcut Constructions for Tree-restricted Shortcuts} 
The main reason for Haeupler, Izumi and Zuzic~\cite{haeupler2016low} to introduce the concept of $b$-block tree-restricted shortcuts is that this additional structure can be used to obtain a very simple and efficient shortcut construction. This construction was further improved by Haeupler, Hershkowitz, and Wajc~\cite{haeupler2018round} to give a slightly faster, deterministic, and message optimal construction:

\begin{lemma}\label{lem:construction}[\cite{haeupler2018round,haeupler2016low}]
There exists a simple distributed algorithm which, for any $n$-node $m$-edge $D$-diameter graph $G$ which admits $c$-congestion $b$-block partial shortcuts with quality $Q = c + bD$, computes a quality-$\tilde{O}(Q)$ shortcut for any given collection of parts in $\tilde{O}(Q)$ rounds with high probability using $\tilde{O}(m)$ message (or deterministically in $\tilde{O}(bQ)$ rounds).
\end{lemma}

By proving that the shortcuts from \Cref{thm:MainExistentialDilation} can be chosen to be tree-restricted with a small $\delta(G)$ block number we get an efficient construction algorithm, and therefore \Cref{thm:mainconstructive}, ``for free''. We remark that proving this additional tree-restriction structure and thus having a fast construction algorithm is crucial for the algorithmic usability of shortcuts and generally quite hard. 
Indeed, efficiently constructing good general shortcuts for all graph families that admit them remains a major open problem. Even in the special case of well-connected graphs (with a small mixing time, e.g., expander or random graphs), for which shortcuts generally cannot be chosen to be tree-restricted, there is currently a $2^{O(\sqrt{\log n})}$ gap between the construction time and the shortcut quality\cite{ghaffari2018mixing, ghaffari2017mixing}.

\section{Shortcuts for Graphs with Minor Density $\delta$}

\subsection{Main Result}

We prove the following main result, which directly implies \Cref{thm:MainExistentialDilation} and \Cref{thm:mainconstructive}. 

\begin{figure}[t]
\begin{center}
\label[figure]{fig:onePart}
\includegraphics[width=0.95\textwidth]{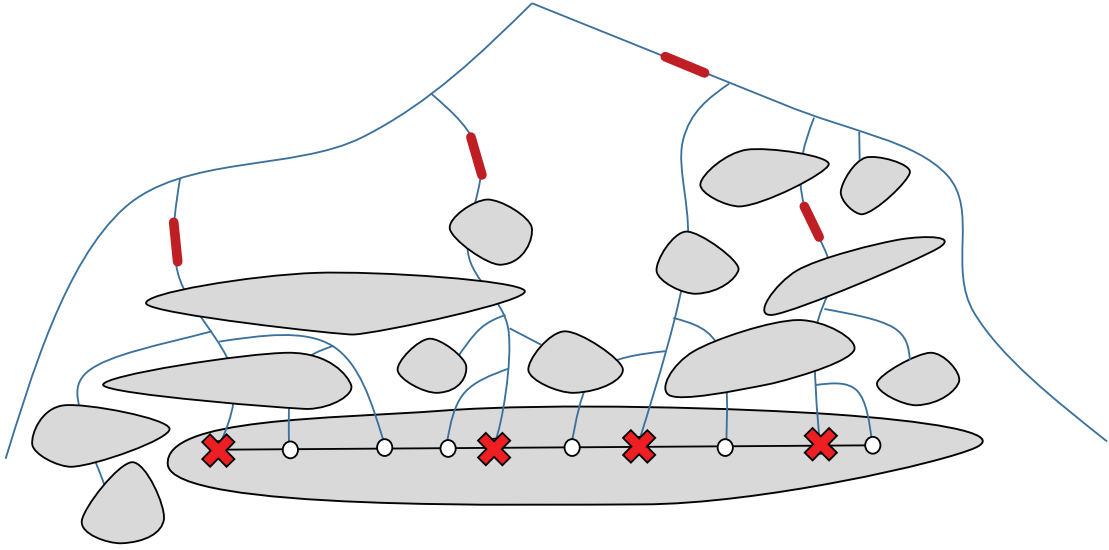}
\end{center}
\caption{\small A schematic illustration of a tree $T$ (indicated in blue), the one part that we are focusing on (indicated as a gray area with white nodes in it) along with a few others (indicated just as gray areas), and the overcongested edges (indicated in red) that have descendants in this part, which we call representatives (indicates with red cross marks). When the blue tree goes through a part (a gray area), we indicate that this section of the tree has some vertex of that part.}
\end{figure}

\begin{theorem}\label{thm:Main}
Every $G$ with diameter $D$ and minor density $\delta = \delta(G)$ admits tree-restricted $8 \delta D$-congestion $8 \delta$-block partial shortcuts.
\end{theorem}

Indeed, using \Cref{obs:PartialImpliesFullShortcut} and \Cref{obs:BlockImpliesDilation} the existence of a $8 \delta D$-congestion $8 \delta$-block partial shortcut directly implies the existence of an $(8\delta D \log_2 n)$-congestion $(8 \delta (2D +1))$-dilation shortcut and therefore \Cref{thm:MainExistentialDilation}. The constructive main theorem \Cref{thm:mainconstructive} directly follows from using \Cref{lem:construction} on the tree-restricted shortcuts of \Cref{thm:Main}.

The general idea to prove \Cref{thm:Main} is to "run" the shortcut construction algorithm from \cite{haeupler2016low} and prove that if it fails to find a sufficiently good tree-restricted shortcut, then $G$ contains a minor with density exceeding $\delta(G)$. 



\begin{proof}[Proof of \Cref{thm:Main}]

Let $T$ be any rooted spanning tree in $G$ of depth at most $D$. Let $\Parts = \{P_1, \ldots, P_k\}$ be a collection of connected node-disjoint parts. We set our desired congestion to be $c = 8 \delta D$. For a tree edge $e \in T$, let $v_e$ be the endpoint of $e$ that is further away from the root.

\paragraph{Defining overcongested edges} Initially, let $O=\emptyset$. We process tree edges in order of decreasing depths, level by level. For any edge $e \in T$, let $I_e \subseteq \Parts$ be the parts that have a non-empty intersection with the descendants of $v_e$ in $T\setminus O$. If $|I_e| \geq c$ we say $e$ is \textit{overcongested} and we add $e$ to $O$. 



\paragraph{The bipartite graph $B$}  We define the bipartite graph $B = (O \cup \Parts, E')$ whose node set consists of edge-nodes corresponding to overcongested edges on the one side and part-nodes corresponding to the parts from $\Parts$ on the other side. The edges $E' = \{(e,P_i) \mid e \in O, P_i \in I_e\} \subseteq O \times \Parts$ of $B$ indicate which part contributed to which edge being overcongested. We associate every edge $(e,P_i) \in B$ with some representative node $r_{(e,P_i)} \in P_i$, that is a descendant of $v_e$ and can be reached from $v_e$ via $T \setminus O$. A schematic illustration of this setup is given in \Cref{fig:onePart}. 
Let $R_i \subseteq P_i$ be the set of representative nodes in $P_i$. Note that $|R_i|$ is equal to the degree of node $P_i \in B$. The degree of any edge-node $e \in B$ is $|I_e| \geq c$, as we only have overcongested edges represented in $B$.

Next, we argue that that one of the following two cases applies: either (I) there exists a good partial shortcut, or (II) graph $G$ has a minor of density exceeding $\delta$  -- contradicting the assumption that $\delta = \delta(G)$.

\paragraph{(I) Either there exists a good partial shortcut}
If at least half of all parts have a degree of at most $8 \delta$ in $B$, then defining $S_i$ for any such part $P_i$ to be all ancestor edges of $P_i$ in the forest $T \setminus \{O\}$ identifies a $c$-congestion $8 \delta$-block partial shortcut.

\paragraph{(II) Or there is a dense minor $B_{\Parts'}$ in $G$}
If we are not in case (I), then at least half of all parts in $B$ have degree at least $8 \delta$ in $B$. In this case, the average degree among part-nodes in $B$ is at least $4 \delta$, because at least half of the parts have degree at least $8\delta$. Moreover, the average degree among edge-nodes (and in fact even their minimum degree) in $B$ is at least $c$.

Let $\Parts'$ be a random subset of $\Parts$ in which each part is included independently at random with probability $\frac{1}{4D}$. We define a subgraph $B_{\Parts'}=(V_{P'}, E_{P'})$ of the bipartite graph $B$ which is also a minor of $G$, as follows. The part-nodes in $B_{\Parts'}$ are exactly the $P_i \in \Parts'$ and in the minor-mapping $\text{map}_{B_{\Parts'},G}$ such a node is mapped to the vertex set of $P_i$. All edges $e \in O$ with $v_e \notin \bigcup_{P_i \in \Parts'} P_i$ are edge-nodes in $B_{\Parts'}$.
The vertex set $\text{map}_{B_{\Parts'},G}(e)$ in $G$ of such an edge-node $e \in B_{\Parts'}$ is exactly the vertices in the connected component containing $v_e$ in the forest $(T \setminus O) \setminus ( \bigcup_{P_i \in \Parts'} P_i)$. 

To define which edges $(e,P_i) \in B$ are in $B_{\Parts'}$, we say $(e,P_i)$ is \textit{potentially present} if the tree path between $v_e$ and the representative $r_{(e,P_i)}\in P_i$, including the deeper endpoint $v_e$ but excluding the representative node $r_{(e,P_i)} \in P_i$, does not contain any node from $\bigcup_{P_j \in \Parts'} P_j$. We say edge $(e,P_i)$ is \textit{actually present} and add $(e,P_i) \in B$ to $B_{\Parts'}$ if it is both potentially present and $P_i \in \Parts'$. Note that $B_{\Parts'}$ is indeed a minor of $G$  under the mapping function $\text{map}_{B_{\Parts'},G}$ since the vertex sets corresponding to nodes in $B_{\Parts'}$ are disjoint and connected in $G$ and edges in $B_{\Parts'}$ are a subset of the edges produced when contracting these vertex sets in $G$.

\paragraph{Density of the graph $B_{\Parts'}=(V_{P'}, E_{P'})$}
Let $k$ be the number of edges in $B$. Every edge $(e,P_i) \in B$ has a probability of at least $1 - (1 - \frac{1}{4D})^D \geq \frac{3}{4}$ to be potentially present. This probability is independent from the $\frac{1}{4D}$ probability for $P_i$ to be in $\Parts'$. Hence, $\mathbb{E}[|E_{P'}|] \geq \frac{3k}{16D}$.
\smallskip

The nodes in $B_{\Parts'}$ on the other hand consist of (A) at most some $\frac{k}{c}$ edge-nodes, given that they have degree at least $c$ in $B$, each of which is included in $B_{\Parts'}$ with probability $1 - \frac{1}{4D}$ and (B) at most some $\frac{k}{4\delta}$ part-nodes, with average degree more than $4\delta$ in $B$, each of which is included in $B_{\Parts'}$ with probability $\frac{1}{4D}$. Hence, 
\[\mathbb{E}[|V_{P'}|] < \frac{k}{8\delta D} + \frac{k}{4\delta} \cdot \frac{1}{4D} = \frac{3k}{16D} \frac{1}{\delta}.
\]
Therefore, by linearity of expectation, we can conclude that 
\[
    \mathbb{E}[|E_{P'}| - \delta |V_{P'}|] = \mathbb{E}[|E_{P'}|] - \delta \mathbb{E}[|V_{P'}|] > \frac{3k}{16D} - \delta \frac{3k}{16D} \frac{1}{\delta} =0,
\] 
which implies that $Pr[|E_{P'}| - \delta |V_{P'}| > 0]>0$. That is, with a non-zero probability\footnote{With a slightly more careful argument, we can show that there is $\Omega(1/D)$ probability to find a minor of density exceeding $\delta$, but for our existence proof, just a positive probability suffices.}, the minor $B_{\Parts'}$ in $G$ has density exceeding $\delta$, giving the desired contradiction.
\end{proof}

We remark that the above proof of \Cref{thm:Main} can easily be made constructive directly. A trivial implementation would lead to an deterministic $O(\delta^2 D^2)$-round algorithm. Using the sampling idea from \cite{haeupler2016low,haeupler2018round} to identify overcongested-edges one can speed this up to $\tilde{O}(\delta D)$. Overall, this would safe a $\Theta(\log n)$ factor in the quality of the computed (partial) shortcut. One could also make the algorithm certifying, i.e, output a dense minor if a (partial) shortcut of desired quality cannot be found. For example, one can obtain an algorithm which when run on a graph $G$ with tree $T$ of depth $D$ and a collection of parts terminates in $\tilde{O}(\delta D)$ rounds for some $\delta \leq \delta(G)$ and outputs both an $8\delta$-block $8\delta D$-congestion partial shortcut and a $(\delta - 1)$-dense bipartite minor, which explains/certifies why no better shortcut was found.

\subsection{Optimality of the Main Result}

Next we prove that our main result \Cref{thm:Main} is existentially optimal up to only small explicitly given \emph{constant} factors in the congestion and block number of partial shortcuts. This directly implies the slightly weaker \Cref{lem:tightnessOfDensityDependecy} for full shortcuts presented in \Cref{sec:OurResults}.

\begin{lemma}\label{lem:LowerBound}
For every $\delta', D' \in N$ with $5 \leq \delta' \leq D'/2$ there exists a topology $G$ (with $O(\delta D)$ nodes) and a set of node-disjoint paths such that:
\begin{itemize}
    \item $G$ has diameter $D'$ and every minor of $G$ has density less than $\delta'$ and
    \item the best partial shortcut quality for the set of paths in $G$ is at least $\frac{(\delta'-3) D'}{6} = \Theta(\delta' D')$, i.e., any partial shortcut has either congestion at least $\Omega(\delta' D')$ or dilation at least $\Omega(\delta' D')$. 
\end{itemize}
\end{lemma}

\begin{figure}[t]
\begin{center}
\label[figure]{fig:lowerbound-minor}
\includegraphics[width=0.95\textwidth]{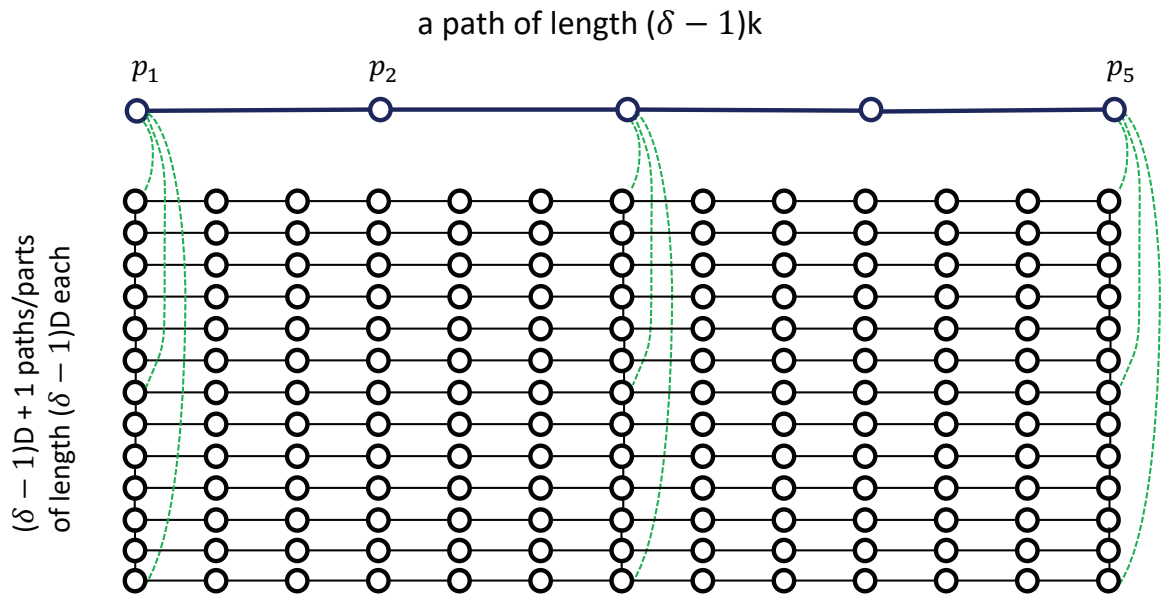}
\end{center}
\caption{\small The lower bound graph for \Cref{lem:LowerBound} with $\delta'=3$, $k=2$ and $D' = 6$. The special path of length at most $D$ is at the top, and below are the $\Theta(\delta D)$ paths/parts of length $\Theta(\delta D)$. Every $D$ steps on these parts there is a column connecting all paths and (green) connections every $D$ steps on this column to the special path on top.}
\end{figure}

\begin{proof}
Let $\delta = \delta'-2$, $k = \floor{\frac{D'}{2\delta}}$, $D=k \delta$.
Note that $k \geq 2$, $\delta' \geq 3$, $D \in [6, \floor{\frac{D'}{2}}]$.
The topology $G=(V, E)$ is made of one special path of length $(\delta-1)k+1$ at the top, along with $(\delta-1) D+1$ many paths of length $(\delta-1) D+1$ at the bottom, known as rows. In every $D^{th}$ column, every $D^{th}$ row is connected to a node in the top path, such that all these nodes of the same column are connected to one node in the top node. See \Cref{fig:lowerbound-minor} for an illustration.
More formally, the graph is defined as follows:
\begin{itemize}
    \item $V = \{p_i \mid i \in [(\delta-1)k+1]\} \cup \{v_{i,j} \mid i,j \in [(\delta-1) D+1]\}$
    \item The edges  $E$ are such that the $p$-nodes form a path of length $(\delta-1)k$. Moreover, for any $i \in [(\delta-1) D+1]$, the vertices $v_{i,*}$-nodes form a path $P_i$ of length $(\delta-1) D$. Also, for any $j \in [\delta] $, the nodes $v_{*,(j-1) D+1}$ form a path of length $(\delta-1) D$ of which every $D^{th}$ node connects to $p_{(j-1)k+1}$, that is, $\{v_{(j'-1) D +1,(j-1) D +1} ,p_{(j-1)k+1}\} \in E$ for every $j,j' \in [\delta]$. 
\end{itemize}

We first argue that the graph has diameter at most $D'$. From every $v$ node one can reach a $p$-node by going at most $\frac{D}{2}$ steps to the closest node in its $P$-path which is in the same column as of a $p$-node, then going at most $\frac{D}{2}$ steps up or down to a node that has a $p$-node neighbor, and then doing one more step to that $p$-node. From any $p$-node one can reach $p_{\frac{D}{2}}$ in at most $\frac{D}{2}$ steps. Overall the diameter of $G$ is therefore at most $1.5D + 1 \leq D'$.

Now we argue that $G$ has no minor with density $\delta'$. Notice that $G$ is planar after deleting the $\delta(\delta-1)$ edges between the $p$ nodes and any $P_i$-path except $P_1$. Any minor of $G$ with $s\geq \delta+1$ nodes therefore has, according to Euler's formula, at most $3s-6+\delta(\delta-1)$ edges and thus an edge density of at most $\frac{3s-6+\delta(\delta-1)}{s} < 3 + \frac{\delta}{s} (\delta-1) < \delta + 2 = \delta'$.

Finally, we argue about the shortcut quality. We have $(\delta-1) D+1$ paths $P_i$, each of length $(\delta-1) D$. For each path $P_i$, the only way to shorten the distance between the two endpoints is to use the edges of the top-path. In fact, unless a part $P_i$ has at least $1/2$ of the edges of the top path in its shortcut $H_i$, the dilation of its part, i.e., the diameter of $G[P_i]+H_i$ would be at least $\frac{1}{2}(\delta-1) D$. Therefore, if the shortcut quality is smaller than $\frac{1}{2}(\delta-1) D$, each part $P_i$ needs to have at least half of the edges of the top path in its shortcut. But then, overall, the $(\delta-1)k+1$ edges of the top path appear at least $\frac{1}{2}((\delta-1) D)((\delta-1)k+1)$ times, in total, in shortcuts. Hence, at least one edge has congestion at least $\frac{1}{2}(\delta-1) D$. Lastly $\frac{1}{2}(\delta-1) D \geq \frac{1}{6} (\delta'-3) D'$.
\end{proof}

We note that there is a $\Theta(\log n)$ factor gap in the congestion (but not the dilation) between the upper bound for full shortcuts in \Cref{thm:MainExistentialDilation} and  \Cref{lem:tightnessOfDensityDependecy}. This gap stems from the $\log_2 n$ loss of $\Cref{obs:PartialImpliesFullShortcut}$, which goes from partial shortcuts to (full) shortcuts. A planar lower bound topology with $\Omega(D \log D)$ shortcuts given in \cite{GH-PlanarII} shows that the congestion-gap between partial shortcuts and full shortcuts can be at least $\Omega(\log D)$. Whether this is the maximal gap between partial and full shortcuts and whether 
\Cref{thm:MainExistentialDilation} can be improved to shortcuts of quality independent of $n$, e.g., $O(\delta D \log D)$, are interesting questions -- albeit not ones of particular importance to the algorithmic applications of shortcuts. 

\subsection{Shortcuts for other Graph Parameters and Algorithmic Applications}

We finish by giving the few remaining technical details for the direct implications of our existential and algorithmic shortcut guarantees from \Cref{thm:Main} and \Cref{thm:mainconstructive}, which are stated in \Cref{sec:OurResults}. 

As discussed in \Cref{sec:OurResults}, our shortcut guarantees in terms of the minor density $\delta(G)$ directly imply bounds for other graph parameters of interest. Consider, for example, the following easy and existentially tight bounds on the minor density in terms of graph parameters considered in \cite{haeupler2016treewidth}:

\begin{lemma}\label{lem:boundsongenusandtreewidth}
The following bounds on the minor density $\delta(G)$ of a graph $G$ hold:
\begin{itemize}
    \item If $G$ has genus, non-orientable genus, or Euler genus of $g$, then $\delta(G) = O(\sqrt{g})$.
    \item If $G$ has treewidth or pathwidth at most $k$, then $\delta(G) \leq k$.
\end{itemize}
\end{lemma}
\begin{proof}
Being embeddable into an orientable or non-orientable surface is a graph property that is closed under minor-operations. Moreover, An $n$-node graph $G$ having a genus, non-orientable genus, or Euler genus of $g$ implies that this graph has at most $3n+O(g)$ edges. Since any graph with density $\delta$ has at least $n \geq \delta$ nodes and thus at least $\min\{\delta n,\frac{\delta^2}{2}\} \geq \frac{\delta n}{2} + \frac{\delta^2}{4}$ edges, the (minor) density of any such graph is at most $O(\sqrt{g})$.
Similarly, having treewidth or pathwidth at most $k$ are graph properties closed under taking minors. Moreover, a graph of treewidth (or pathwidth) at most $k$ and $n$ nodes has less than $kn$ edges and therefore its (minor) density is at most $k$. 
\end{proof}

Combining \Cref{thm:Main} with \Cref{lem:boundsongenusandtreewidth} now directly implies \Cref{cor:genus} and the following analogous corollary for treewidth-$k$ graphs:

\begin{corollary}\label{cor:treewidth}
Any graph with $n$ nodes, diameter $D$, and treewidth at most $k$ admits shortcuts with congestion $O(k D \log n)$ and dilation $O(kD)$. 
\end{corollary}

Note that while completely different and highly nontrivial proofs specific to planar graphs, bounded genus graphs, and bounded treewidth graphs were given in \cite{GH-PlanarII, haeupler2016treewidth}, we obtain the same existentially optimal results by simply plugging in bounds on the minor density $\delta(G)$ in terms of the desired graph parameter into \Cref{thm:Main}. Constructive results of these shortcuts follow similarly from \Cref{thm:mainconstructive}.

Algorithmic applications, such as, the fast distributed MST and minimum-cut algorithms claimed in \Cref{cor:MST} and \Cref{cor:MinCut} follow immediately and in a completely modular fashion from our new constructive shortcuts given in \Cref{thm:mainconstructive} and shortcut-based algorithms like the ones given in~\cite{GH-PlanarII,haeupler2018round}. 

\begin{proof}[Proof of \Cref{cor:MST} and \Cref{cor:MinCut}]
The min-cut algorithm follows from a randomized algorithm given in \cite{GH-PlanarII} that computes a $(1+\eps)$ approximation in $\tilde{O}(Q\poly(1/\eps))$ time, with high probability, given a shortcut of quality $Q$ which is can be constructed according to  \Cref{thm:mainconstructive}. To convert this into an exact algorithm we observe that the minimum degree of $G$ and therefore also its min-cut is of size at most $2\delta$ given that the density of $G$ can be at most $\delta$. Setting $\eps = \frac{1}{4\delta}$ therefore implies an exact algorithm.
The MST statement follows directly from \Cref{thm:mainconstructive} and Boruvka's distributed MST algorithm, as described in~\cite{GH-PlanarII,haeupler2018round}.
\end{proof}

\section*{Acknowledgements}
The second author thanks Anupam Gupta for suggesting the question that led to this paper.

\newcommand{\etalchar}[1]{$^{#1}$}


\end{document}